\documentclass{article}
\usepackage{fullpage}
 
\usepackage{microtype}
\usepackage{amsmath}
\usepackage{amsthm}
\usepackage{amssymb}
\usepackage[ruled,vlined,linesnumbered]{algorithm2e}
\usepackage{tikz}
\usepackage{enumitem}
\usepackage{wrapfig}
\usepackage{cutwin}
\usepackage{float}
\usepackage{subfig}
\usepackage{graphicx,import}
\usepackage[hidelinks]{hyperref}

\newcommand{\ee}{\varepsilon}
\newcommand{\NN}{\mathbb N}
\newcommand{\ZZ}{\mathbb Z}
\newcommand{\RR}{\mathbb R}
\newcommand{\man}{m}
\newcommand{\manStrat}{M}
\newcommand{\goal}{g}
\newcommand{\lion}{l}
\newcommand{\deltaT}{\sigma}
\newcommand{\rhoD}{\rho}
\newcommand{\rD}{r}
\DeclareMathOperator{\dist}{dist}
\newcommand{\mydef}{:=}

\newcommand{\deriv}{\delta}

\usepackage{amsthm}
\newtheorem{theorem}{Theorem}
\newtheorem{claim}[theorem]{Claim}

\newtheorem{lemma}[theorem]{Lemma}

\title{Escaping an Infinitude of Lions}
\author{Mikkel Abrahamsen\thanks{Basic Algorithms Research Copenhagen (BARC), University of Copenhagen, Denmark.
$\{$\texttt{miab},\texttt{jaho},\texttt{koolooz}$\}$\texttt{@di.ku.dk}.
MA was supported by Advanced Grant DFF-0602-02499B from the Independent Research Fund Denmark under the Sapere Aude research career programme.
CWN was supported by the Starting Grant DFF-7027-00050B from the Independent Research Fund Denmark under the Sapere Aude research career programme.
}
\and Jacob Holm$^*$ \and 
Eva Rotenberg\thanks{Technical University of Denmark. \texttt{erot@dtu.dk}
} \and
Christian Wulff-Nilsen$^*$
}

\makeatletter
\let\copyrightline\@empty
\makeatother

\theoremstyle{theorem}

\date{} 

\begin{document}

\maketitle

\begin{abstract}
We consider the following game played in the Euclidean plane:
There is any countable set of unit speed lions and one fast man who can run with speed $1+\ee$ for some value $\ee>0$.
Can the man survive?
We answer the question in the affirmative for any $\ee>0$.
\end{abstract}

\paragraph{Note:}
A preliminary version of this paper was part of the paper~\cite{abrahamsen2017best}.
After the publication of this paper~\cite{abrahamsen2020escaping}, it has come to our attention that Chernous'ko~\cite{chernousko1976problem} proved already in 1976 that a fast man can escape any finite number of lions.
It seems that the same technique can be used to escape from an infinite set.
For work on a related game and more references, see the paper by Ibragimov, Salimi, and Amini~\cite{ibragimov2012evasion}.

\section{Introduction}
``A lion and a man in a closed circular arena have equal maximum speeds.
What tactics should the lion employ to be sure of his meal?''\footnote{The
curve of pursuit ($L$ running always straight at $M$) takes infinite time,
so the wording has its point.} These words (including the footnote) introduce the
now famous lion and man problem, invented by R.~Rado in the
late 1930s, in Littlewood's
\emph{Miscellany}~\cite{littlewood1986littlewood}.

It was for a long time believed that
in order to avoid the lion,
it was optimal for the man to run on the boundary of the arena. A simple
argument then shows that the lion
could always catch the man by staying on the radius containing the
man while approaching him as much as possible.
However, A. S.~Besicovitch
proved in 1952 that the man has a very simple strategy
(following which he will approach but not reach the boundary)
that enables him to avoid capture forever no matter
what the lion does. See~\cite{littlewood1986littlewood} for details.

One can prove that
two lions are enough to catch the man in a circular arena,
and Croft~\cite{croft1964lion}
proved that in general a necessary and sufficient number of 
birds
to catch a
fly inside an $n$-dimensional spherical cage is just
$n$ (again, we assume that the fly and the birds are points with equal maximum speeds).

Rado and Rado~\cite{rado1974more} and
Jankovi{\'c}~\cite{jankovic1978about}
considered the problem where there are many lions
and one man, but where the game is played in the entire unbounded plane.
They proved that the lions can catch the man
if and only if the man starts in the interior of the convex hull of
the lions. Inspired by that problem,
we ask the following question: What if the lions have maximum speed
$1$ and the man has maximum speed $1+\ee$ for some $\ee>0$?
We prove that for any $\ee >0$ and any countably infinite set of lions, such a fast man can survive forever provided that he does not start at the same point as one of the lions.
We find this result surprising.
Indeed, it is difficult to imagine how that man proceeds if, say, he starts at the point $(\sqrt 2,0)$ and there are lions at all points with two rational coordinates.
In Section~\ref{sec:mainThm}, we state the theorem and show how it follows from a technical lemma.
In Section~\ref{sec:techLemma}, we prove the technical lemma.

Lion and man games with a fast man have been considered previously.
Ramana and Kothari~\cite{ramana2017pursuit} and Bakolas~\cite{bakolas2017finite} studied variants with one or more lions that catch the man if he is within a certain positive distance from a lion.
Here, the movements of the man and the lions are also restricted in other ways than just limiting the speeds.
Flynn~\cite{flynn1973lion,flynn1974lion} and
Lewin~\cite{lewin1986lion} studied the problem where there is one lion and one fast man
in a circular arena. The lion tries to get as close to the man as possible and the
man tries to keep the distance as large as possible.
Variants of the cop and robber game (a discrete version of the lion and man game played on graphs) where the robber is faster than the cops have also been studied. See for instance~\cite{alon2015chasing,fomin2010pursuing}.

\subsection{Definitions}\label{definitions}

We follow the conventions of Bollob{\'a}s et al.~\cite{bollobas2012lion} for games with one lion and one man.
Let $R\subseteq\RR^2$ be a region in the plane on which the game is to be played, and assume that the lion starts at point $l_0$ and the man
at point $m_0$.
We define a \emph{man path} to be a function $m\colon [0, \infty)\longrightarrow R$ satisfying $m(0)=m_0$ and
the Lipschitz condition
$\lVert m(s)-m(t)\rVert\leq (1+\ee)\cdot \lvert s-t\rvert$ for some small
$\ee>0$. Note that it follows from the Lipschitz condition
that any man path is continuous.
A \emph{lion path} $l$ is defined similarly, but with $\ee=0$, i.e., the lion always runs
with (at most) unit speed.
Let $\mathcal L$ be the set of all lion paths and $\mathcal M$ be the set of all
man paths. Then a \emph{strategy} for the man
is a function $M\colon \mathcal L\longrightarrow\mathcal M$ such that
if $l,l'\in\mathcal L$ agree on $[0,t)$, then $M(l)$ and $M(l')$ also agree on
$[0,t]$. This last condition is a formal way to describe that the man's position
$M(l)(t)$, when he follows strategy $M$, depends only on
the position of the lion at points in time
before time $t$, i.e., he is not allowed to act based on the lion's future
movements.
(By the continuity of the paths, it does not matter whether these time intervals are open or closed.)
A strategy $M$ for the man is \emph{winning} if for any $l\in\mathcal L$ and
any $t\in[0,\infty)$, we have that $M(l)(t)\neq l(t)$.
Similarly, a strategy for the lion $L\colon\mathcal M\longrightarrow\mathcal L$
is winning if for any $m\in\mathcal M$, we have that $L(m)(t)=m(t)$
for some $t\in[0,\infty)$.
These definitions are extended to games with more than one lion in the natural way.

It might seem unfair that the lion is not allowed to react on the man's movements
when we evaluate whether a strategy $M$ for the man is winning. However,
we can give the lion full information about $M$ and
allow it to choose its path $l$ depending on $M$ \emph{prior}
to the start of the game. If $M$ is a winning strategy, the man can also
survive the lion running along $l$.

We call a strategy of the man $M$ \emph{locally finite} if it satisfies the following property:
if $l$ and $l'$ are any two
lion paths that agree on $[0, t]$ for some $t$, then the corresponding man paths
$M(l)$ and $M(l')$ agree on $[0, t+\delta]$ for some
$\delta > 0$. (We allow $\delta$ to depend on $l\vert_{[0,t]}$.)
Informally, a strategy is locally finite if at any point in time the man commits to doing something for some positive
amount of time in the future dependent only on the situation in the past.
In particular, the strategy is \emph{not} locally finite if it involves staying at a fixed distance from the lion or other behaviors that require the man to react instantly based on the lion's actions.

Bollob{\'a}s et al.~\cite{bollobas2012lion} proved that if the man has a locally finite
winning strategy, then the lion does not have any winning strategy.
The argument easily extends to
games with multiple lions.
At first sight, it might sound absurd to even consider the possibility that the lion
has a winning strategy when the man also does. However, it does not follow
from the definition that the existence of a winning strategy for the man
implies that the lion does not also have a winning strategy.
See the paper by Bollob{\'a}s et al.~\cite{bollobas2012lion} for a detailed discussion
of this (including descriptions of variants of the lion and man game where
both players have winning strategies). The winning
strategy of the fast man against finitely many lions is locally finite, so it follows that the lions do not
have a winning strategy.
In fact, the man's strategy satisfies the stronger
condition that it is \emph{equitemporal}, i.e., there is a $\deltaT>0$ such that the man at any point in time $i\cdot\deltaT$, for $i\in\mathbb N_0$,
decides his behavior until time $(i+1)\cdot\deltaT$.
This is a special case of a locally finite strategy since at time $t\in [i\cdot\deltaT,(i+1)\cdot\deltaT)$, the man commits to doing something in the future interval $(t,(i+1)\cdot\deltaT]$.
However, as the number of lions tends to infinity, $\deltaT$ tends to $0$, so the winning strategy against infinitely many lions is not locally finite, and the lions might have a winning strategy as well.

\section{Main Theorem}\label{sec:mainThm}

\begin{theorem}\label{thm:main}
    In the Euclidean plane, for any $\ee>0$, let a man able to run at speed $1+\ee$ start at a point $\man_0$ and a unit speed lion start at a point $\lion_i(0)\neq\man_0$ for each $i\in\NN$.
    Then the man has a winning strategy against these infinitely many lions.
\end{theorem}

In the rest of the article, we adopt the setting of Theorem~\ref{thm:main}.
We use $\manStrat_n$ to denote a strategy (to be defined later) of the man against $n$ lions starting at the points $\lion_1(0),\ldots,\lion_n(0)$.
We assume that we are given arbitrary lion paths $\lion_1,\lion_2,\ldots$ such that $\lion_i$ starts at the point $\lion_i(0)$.
With slight abuse of notation, we also use $\manStrat_n$ to denote the man path prescribed by strategy $\manStrat_n$ against the lion paths $\lion_1,\ldots,\lion_n$, that is, $\manStrat_n$ is shorthand for the formally correct notation $\manStrat_n(\lion_1,\ldots,\lion_n)$.

In order to prove the theorem, we make use of the following technical lemma, the proof of which is deferred to Section~\ref{sec:techLemma}.

\begin{lemma}\label{LEMMA:MAIN}
In the Euclidean plane, for any $\ee>0$, let a man able to run at speed $1+\ee$ start at a point $\man_0$ and a unit speed lion start at a point $\lion_i(0)\neq\man_0$ for each $i\in\NN$.
Let a sequence $(\deriv_i)_{i\geq 2}$ of positive real numbers be given.
For every $n\in\NN$, there is a strategy $\manStrat_n$ for the man against $n$ lions starting at points $\lion_1(0),\ldots,\lion_n(0)$ with the following properties:
\begin{enumerate}[leftmargin=1.6em]
\item
There is a \emph{safety distance} $c_n>0$ such that while following $\manStrat_n$, the man maintains distance at least $c_n$ from lion $\lion_n$. \label{LEMMA:MAIN:PROP0}

\item
If $n\geq 2$, then $\|\manStrat_{n-1}(t)-\manStrat_n(t)\|\leq \deriv_n$ at any time $t$. \label{LEMMA:MAIN:PROP0a}

\item
$\manStrat_n$ depends on $\deriv_i$ only for $i\in\{2,\ldots,n\}$. \label{LEMMA:MAIN:PROP1}
\end{enumerate}
\end{lemma}

We now show how Theorem~\ref{thm:main} follows from the lemma.

\begin{proof}[Proof of Theorem~\ref{thm:main}.]
We define a sequence $(\deriv_i)_{i\geq 2}$ inductively, as follows.
Suppose that for some $n\geq 2$, we have already defined $\delta_i$ for all $i\in\{2,\ldots,n-1\}$.
By Lemma~\ref{LEMMA:MAIN} (note in particular the use of property~\ref{LEMMA:MAIN:PROP1}), these values of $\delta_i$ yield strategies $\manStrat_1,\ldots,\manStrat_{n-1}$ with associated safety distances $c_1,\ldots,c_{n-1}$.
We then define
$$\deriv_n\mydef\min\left\{1/2^n,\min_{i\in\{1,\ldots,n-1\}}\frac{c_i}{2^{n-i+1}}\right\}.$$

We now claim that the resulting sequence of strategies $\manStrat_1,\manStrat_2,\ldots$ converges to a winning strategy $\manStrat_\infty$.
Note that for any $n,m\in\NN$, where $n<m$, we have
$$
\|\manStrat_n(t)-\manStrat_m(t)\|\leq \sum_{i=n}^{m-1}\|\manStrat_i(t)-\manStrat_{i+1}(t)\|\leq \sum_{i=n}^{m-1}1/2^{i+1}<1/2^n.
$$
Hence, the sequence $(\manStrat_n(t))_{n\in\NN}$ is a Cauchy sequence and converges to some point $\manStrat_\infty(t)$.

We need to ensure that $M_\infty(t)$ moves with speed at most $1+\ee$.
This is indeed the case since for any two points in time $s,t$, we have for any $n$ that $\|M_n(s)-M_n(t)\|\leq(1+\ee)\cdot|s-t|$.
Therefore, it must also be the case that $\|M_\infty(s)-M_\infty(t)\|\leq(1+\ee)\cdot|s-t|$.

Finally, we need to verify that $\manStrat_\infty$ is winning.
To this end, we see that for any lion $\lion_i$ and strategy $\manStrat_n$ where $i\leq n$, we have
\begin{align*}
\|\manStrat_n(t)-\lion_i(t)\|\geq{} & \|\manStrat_i(t)-\lion_i(t)\|-\sum_{j=i}^{n-1}\|\manStrat_j(t)-\manStrat_{j+1}(t)\| \\
\geq{} & c_i-\sum_{j=i+1}^{n}\deriv_j\geq c_i\left(1-\sum_{j=1}^{n-i} 1/2^{j+1}\right)>c_i/2.
\end{align*}
It follows that $\manStrat_n$ is winning against the lions $\lion_1,\ldots,\lion_n$.
Furthermore, as the man maintains distance at least $c_i/2$ from $\lion_i$ using any strategy $\manStrat_n$, $i\leq n$, we obtain that the limiting strategy $\manStrat_\infty$ maintains the same distance, so $\manStrat_\infty$ is indeed winning against all the lions.
\end{proof}

As a side remark, we note that it follows from the proof that using the strategy $M_n$ against a finite number $n$ of lions, the man is able to maintain the distance $d_n\mydef \min\{c_1/2,\ldots,c_n/2\}$ from every lion $\lion_1,\ldots,\lion_n$.
Thus, if the $n$ lions and man are disks with radius less than $d_n/2$, the man is still able to win.

\section{Proof of Lemma~\ref{LEMMA:MAIN}}\label{sec:techLemma}





We first give a high-level description of the proof.
We proceed by induction on $n$.
The strategy $\manStrat_n$ yields a curve consisting of line segments all of the same length.

Let $\ee_n\mydef (1-2^{-n})\cdot\ee$ so that $\ee_n<\ee$ for all $n$ and $0<\ee_1<\ee_2<\cdots$.
We define the strategy $\manStrat_n$ such that the man runs at constant speed $1+\ee_n$ when following $\manStrat_n$.
Inductively, we know the previous strategy $\manStrat_{n-1}$ by which the man runs at speed only $1+\ee_{n-1}$.
We place \emph{milestones} at the curve defined by $\manStrat_{n-1}$ equidistantly at distance at most $\deriv_n/2$.
When using strategy $\manStrat_n$, the man runs toward the milestones one by one, so that at any time $t$, he is close to the point $\manStrat_{n-1}(t)$ where he would be when following strategy $\manStrat_{n-1}$.

Using strategy $M_n$, the man runs with speed $1+\ee_n$, i.e., slightly faster than when using strategy $M_{n-1}$.
This gives time to make some detours caused by the lion $\lion_n$ while still being close to the milestone prescribed by strategy $M_{n-1}$.
If $\lion_n$ gets too close, the man makes an \emph{avoidance move}, keeping a safety distance $c_n$ from $\lion_n$.
Intuitively, when performing avoidance moves, the man runs counterclockwise around a circle of radius $c_n$ centered at the lion $\lion_n$.
After a limited number of avoidance moves, the man can make an \emph{escape move}, where he simply runs toward the milestone defined by the strategy $M_{n-1}$.

By choosing $c_n$ sufficiently small, we can make sure that the detour caused by the lion $\lion_n$ is so small that it can only annoy the man once in between any neighboring pair of milestones.
It is then possible to ensure that he will be within distance $\deriv_n$ from the position defined by strategy $M_{n-1}$ at any time, as required by the lemma.

\begin{proof}[Proof of Lemma~\ref{LEMMA:MAIN}]
We assume without loss of generality that $\ee<1$.
We define strategies $\manStrat_1,\manStrat_2,\ldots$ such that $\manStrat_n$ is a strategy for a man starting at point $\man_0$ against $n$ lions starting at points $\lion_1(0),\ldots,\lion_n(0)$.
We define these strategies so that $\manStrat_n$ has the following properties in addition to the properties~\ref{LEMMA:MAIN:PROP0}--\ref{LEMMA:MAIN:PROP1} stated in the lemma:
\begin{enumerate}[leftmargin=1.6em]
\setcounter{enumi}{3}
\item \label{prop:2}
The man is always running at speed $1+\ee_n$, where $\ee_n\mydef(1-2^{-n})\cdot\ee$.

\item \label{prop:3}
Let $t_i\mydef i\cdot \deltaT_n$ for some $\deltaT_n$ to be defined.
The path defined by $\manStrat_n$ is a polygonal chain with corners
$\manStrat_n(t_0),\manStrat_n(t_1),\ldots$ and
each segment has the same length $\|\manStrat_n(t_i)-\manStrat_n(t_{i+1})\|=\deltaT_n\cdot(1+\ee_n)$.

\item \label{prop:3a}
The point $\manStrat(t_{i+1})$ can be determined from the positions of the man and the lions $\lion_1,\ldots,\lion_n$ at time $t_i$.

\end{enumerate}

Properties~\ref{prop:2}--\ref{prop:3a} imply that the strategy is equitemporal and hence locally finite, as explained in the Introduction.
We prove the statement by induction on $n$.
For $n=1$, the man will run on the same ray all the time with constant speed $1+\ee_1=1+\ee/2$.
The man runs directly away from the lion, i.e., in direction $\man_0-\lion_1(0)$, and we define $\deltaT_1\mydef 1$.
This strategy obviously satisfies the stated properties.
Assume now that a strategy $M_{n-1}$ with the stated properties has been defined.


The strategy $\manStrat_{n-1}$ defines a path consisting of segments of length $\ell\mydef \deltaT_{n-1}(1+\ee_{n-1})$.
On each segment we place \emph{milestones} equidistantly (including at the endpoints) so that the segment is divided into $p\mydef \left\lceil\frac{\ell}{\deriv_n/2}\right\rceil$ pieces.
The length of each piece is then $\ell/p\leq \deriv_n/2$.
At any time $t$, the succeeding milestone that the man would pass when following the strategy $\manStrat_{n-1}$ is the point
$$\goal(t)\mydef \manStrat_{n-1}\left(\left\lfloor \frac t{\deltaT_{n-1}/p}+1\right\rfloor\cdot\deltaT_{n-1}/p\right).$$
By property~\ref{prop:3a}, the man can at any time $t$ compute the point
$\goal(t)$.

We now describe the strategy $\manStrat_n$, where the man also has to avoid the lion $\lion_n$.
We first describe the intuition behind the strategy without specifying all details, and~later give a precise description. 

Informally, the strategy $\manStrat_n$ is to attempt to run as according to the strategy $\manStrat_{n-1}$.
Thus, at any time $t$, the man's goal is to run toward the point $\goal(t)$.
However, the lion $\lion_n$ might prevent him from doing so.
Compared to $\manStrat_{n-1}$, the man has now increased his speed by $1+\ee_n-(1+\ee_{n-1})=2^{-n}\ee$, so he has time to take detours while still following the strategy $\manStrat_{n-1}$ approximately.
Note that the goal $\goal(t)$ that the man is attempting to reach will change at any point in time $t$ that is a multiple of $\deltaT_{n-1}/p$.

Assume that we have defined the man's strategy $\manStrat_n$ up to time $t$.
We use $\man(t)$ as shorthand for $\manStrat_n(t)$.
If he is close to the lion $\lion_n$, i.e., the distance $\|\man(t)-\lion_n(t)\|$ is close to $\rD$, for some small constant $\rD>0$ to be specified later, then he runs counterclockwise
around the lion, maintaining approximately distance $\rD$ to the lion.
He does so until he gets to a point where running directly toward $\goal(t)$ will not
decrease his distance to the lion.
He then escapes from the lion $\lion_n$, running directly toward $\goal(t)$.
Doing so, he can be sure that the lion cannot disturb him anymore until he reaches $\goal(t)$ or $\goal(t)$ has changed.

In order to bound the deviation between $\man(t)$ and $\manStrat_{n-1}$, and thus obtaining property~\ref{LEMMA:MAIN:PROP0a}, we choose $\rD$ so small that when the man is running around the lion $\lion_n$, we are in one of the following cases:
\begin{itemize}[leftmargin=1.6em]
\item
The lion is so close to $\goal(t)$ that the man is also within distance $\deriv_n/2$ from $\goal(t)$.

\item
After running around the lion in a period of time no longer than $6\pi \rD/\ee_n$, the man escapes by running directly toward $\goal(t)$ without decreasing the distance to the lion.
By choosing $\rD$ sufficiently small, we can therefore limit the duration, and hence the length, of the detour that the lion can force the man to run, so that the man is ensured to be within distance $\deriv_n$ from the point $\manStrat_{n-1}(t)$ at any time $t$.
\end{itemize}

We now describe the details that make this idea work.
We define
$$\rD\mydef\min\left\{
\frac{\deltaT_{n-1}/p\cdot \ee_n(\ee_n-\ee_{n-1})}{2+2\ee_n+18\pi(1+\ee_n)},
\frac{\deriv_n/2\cdot \ee_n}{2+2\ee_n+12\pi(1+\ee_n)},
\|m_0-\lion_n(0)\|
\right\},$$
$$\rhoD\mydef 2\rD/\ee_n,$$
$$\theta\mydef\arccos\frac 1{1+\ee_n},$$
$$\varphi\in(0,\pi/2]\quad\text{so that}\quad\tan\theta=\frac{\rhoD \sin\varphi}{\rhoD\cos\varphi-2\rD},\quad\text{and}$$
\begin{align*}
\deltaT_n>0 \quad\text{sufficiently small that} \quad  & 2\arcsin\frac{(2+\ee_n)\deltaT_n}{2(\rD-\deltaT_n)}+\frac{\deltaT_n}{\rhoD}\leq\varphi\quad\text{and}\\
& \deltaT_n<\frac \rD{3+\ee_n}. 
\end{align*}

We note that $\varphi$ can be chosen since the function
$x\longmapsto \frac{\rhoD \sin x}{\rhoD\cos x-2\rD}$ is $0$ for $x=0$ and
tends to $+\infty$ as $\rhoD\cos x$ decreases to $2r$.
As for $\deltaT_n$, the function
$x\longmapsto 2\arcsin\frac{(1+\ee_n)x}{2(\rD-x)}+\frac{x}{\rhoD}$
is $0$ for $x=0$ and increases continuously, and hence $\deltaT_n$ can be chosen.

\begin{figure}
\centering
\begin{minipage}[t]{.3\textwidth}
  \centering
  \includegraphics[page=1]{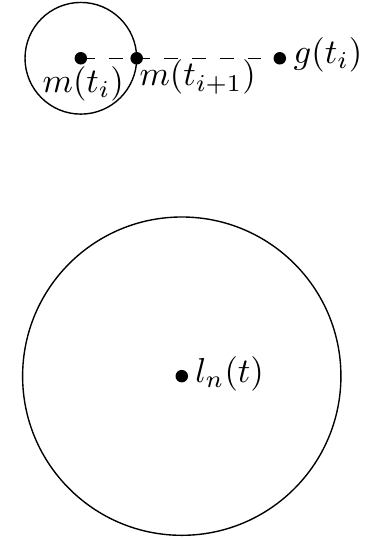}
  \captionof{figure}{A free move. The circles with centers
  $\man(t_i)$ and $\lion_n(t_i)$ have radii $(1+\ee_n)\deltaT_n$ and
  $\rD$, respectively.}
  \label{freeMove}
\end{minipage}\hspace{0.04\textwidth}%
\begin{minipage}[t]{.3\textwidth}
  \centering
  \includegraphics[page=2]{ipeFigs.pdf}
\captionof{figure}{An escape move. The man runs to $b$.}
\label{escapeMove}
\end{minipage}\hspace{0.04\textwidth}
\begin{minipage}[t]{.3\textwidth}
  \centering
  \includegraphics[page=3]{ipeFigs.pdf}
  \captionof{figure}{An avoidance move. The man runs to $q$.}
  \label{avoidanceMove}
\end{minipage}%
\end{figure}

Define a point in time $t$ to be a \emph{time of choice} if
$t$ has the form $t_i\mydef i\deltaT_n$ for $i\in\NN_0$.
At any time of choice $t_i$, the man chooses the point $\man(t_{i+1})$ at distance $(1+\ee_n)\deltaT_n$ from his current position $\man(t_i)$, as given by the following strategy:
\begin{enumerate}[leftmargin=1.6em,label=\textbf{\Alph*}]
\item\label{choice:a}
Suppose first that $\|\man(t_i)-\lion_n(t_i)\|\geq \rD+\deltaT_n(1+\ee_n)$.
Then the man chooses the
direction directly toward $\goal(t_i)$.
In the exceptional case that $\man(t_i)=\goal(t_i)$, i.e., if the man is standing at the goal he is pursuing, he chooses an arbitrary direction.\footnote{It might seem counterintuitive to run away from $\goal(t_i)$ in this case instead of, say, stay in place until $t_{i+1}$. However, it makes the subsequent analysis simpler that the man is always moving at full speed $1+\ee_n$.}
See Figure~\ref{freeMove}.

\item\label{choice:b}
Suppose now that $\|\man(t_i)-\lion_n(t_i)\|< \rD+\deltaT_n(1+\ee_n)$ and
$\man(t_i)\neq \goal(t_i)$.
Let $b$ be the point at distance $(1+\ee_n)\deltaT_n$ from $\man(t_i)$ in the direction
toward $\goal(t_i)$.
If there exist two parallel lines $W_0$ and $W_1$ such that
$\man(t_i)\in W_0$, $b\in W_1$, $\dist(\lion_n(t_i),W_0)\geq \rD-\deltaT_n$,
and $\dist(\lion_n(t_i),W_1)\geq \dist(\lion_n(t_i),W_0)+\deltaT_n$,
then the man runs to $b$.
The man will increase his distance from the lion in a step of this type.
See Figure~\ref{escapeMove}.

\item\label{choice:c}
In the remaining cases,
the circles $C(\man(t_i),\deltaT_n(1+\ee_n))$
and $C(\lion_n(t_i),\rD)$ intersect at two points $p$ and $q$ such that
the arc on $C(\lion_n(t_i),\rD)$ from $p$ counterclockwise to $q$ is in the interior of
$C(\man(t_i),\deltaT_n(1+\ee_n))$. The man then runs toward the point
$q$.
See Figure~\ref{avoidanceMove}.
\end{enumerate}

A move defined by case~\ref{choice:a},
\ref{choice:b}, or \ref{choice:c} is
called a
\emph{free move}, an \emph{escape move}, or an \emph{avoidance move},
respectively.
Let \emph{move $i$} be the move that the man makes during the interval
$[t_i,t_{i+1})$.

\begin{claim}\label{claim1}
At any time of choice $t_i$, we have that
$\|\man(t_i)-\lion_n(t_i)\|\geq \rD-\deltaT_n$
and if the preceding move was an avoidance move, we also have that
$\|\man(t_i)-\lion_n(t_i)\|\leq \rD+\deltaT_n.$
Furthermore, at an arbitrary point in time $t\in [t_{i-1},t_i]$ and any point
$\man'\in \man([t_{i-1},t_i])$ we have that
$0<\rD-(3+\ee_n)\deltaT_n\leq\|\man'-\lion_n(t)\|$
and
if move $i-1$ was an avoidance move then additionally
$\|\man'-\lion_n(t)\|\leq \rD+(3+\ee_n)\deltaT_n.$
\end{claim}

\begin{proof}
We prove the claim by induction on $i$. It clearly holds for $i = 0$ so assume that the claim holds for $i-1$. If move $i-1$ was a free move, then we have
\begin{align*}
\|\man(t_i)-\lion_n(t_i)\|&\geq
\|\man(t_{i-1})-\lion_n(t_{i-1})\|-(2+\ee_n)\deltaT_n\\
&\geq
\rD+\deltaT_n(1+\ee_n)-(2+\ee_n)\deltaT_n=\rD-\deltaT_n.
\end{align*}

If move $i-1$ was an escape move, then let $W_0$ and $W_1$ be lines containing $\man(t_{i-1})$ and $\man(t_i)$, respectively, as defined in~\ref{choice:b}.
Let $b'$ be the intersection of $W_0$ and the segment $\lion_n(t_{i-1})\man(t_i)$.
We then have
\begin{align*}
\|\man(t_i)-\lion_n(t_{i-1})\| & =
\|\man(t_i)-b'\|+\|b'-\lion_n(t_{i-1})\|\\
& \geq \dist(W_0,W_1)+\dist(\lion_n(t_{i-1}),W_0) \\
& \geq \deltaT_n+(\rD-\deltaT_n)=\rD.
\end{align*}
Hence
$$
\|\man(t_i)-\lion_n(t_i)\|\geq \|\man(t_i)-\lion_n(t_{i-1})\|-\deltaT_n\\
\geq \rD-\deltaT_n.
$$

If move $i-1$ was an avoidance move, we have
$$\|\man(t_i)-\lion_n(t_i)\|\geq \|\man(t_i)-\lion_n(t_{i-1})\|-\deltaT_n
= \rD-\deltaT_n
$$
and, similarly,
$$\|\man(t_i)-\lion_n(t_i)\|\leq \|\man(t_i)-\lion_n(t_{i-1})\|+\deltaT_n
= \rD+\deltaT_n.
$$

Since at a point of choice $t_{i-1}$ we have $\rD-\deltaT_n\leq\|\man(t_{i-1})-\lion_n(t_{i-1})\|$ and the lion and the man can move at most $(2+\ee_n)\deltaT_n$ closer to each other within $\deltaT_n$ time, we have for any point in time $t\in[t_{i-1},t_i]$ and any point $\man'\in \man([t_{i-1},t_i])$ that
\begin{align*}
\rD-(3+\ee_n)\deltaT_n\leq
\|\man(t_{i-1})-\lion_n(t_{i-1})\|-(2+\ee_n)\deltaT_n\leq
\|\man'-\lion_n(t)\|.
\end{align*}

If move $i-1$ was an avoidance move, we have $\|\man(t_{i-1})-\lion_n(t_{i-1})\|\leq \rD+\deltaT_n$, and hence for any point in time $t\in[t_{i-1},t_i]$ and any point $\man'\in \man([t_{i-1},t_i])$ that
\begin{align*}
\|\man'-\lion_n(t)\|
\leq
\|\man(t_{i-1})-\lion_n(t_{i-1})\|+(2+\ee_n)\deltaT_n
\leq
\rD+(3+\ee_n)\deltaT_n.
\end{align*}

\end{proof}

One might fear that the lion $\lion_n$ could repeatedly disturb the man and force him to make avoidance moves so that he would never approach $\goal(t)$.
Luckily, the following claim states that if he is able to make an escape move, he will not be distrubed again until he reaches $\goal(t)$ or $\goal(t)$ changes.

\begin{claim}\label{claim2}
An avoidance move is succeeded by an avoidance move or an escape
move. When the man makes an escape move, he will not make an avoidance move before
he reaches $\goal(t)$ or $\goal(t)$ changes.
\end{claim}

\begin{proof}
Consider move $i$.
We know from Claim~\ref{claim1} that if move $i-1$ was an avoidance move, then
$\|\man(t_i)-\lion_n(t_i)\|\leq \rD+\deltaT_n<\rD+(1+\ee_n)\deltaT_n$,
so move $i$ cannot be a free move.

\begin{figure}
\centering
\includegraphics[page=5]{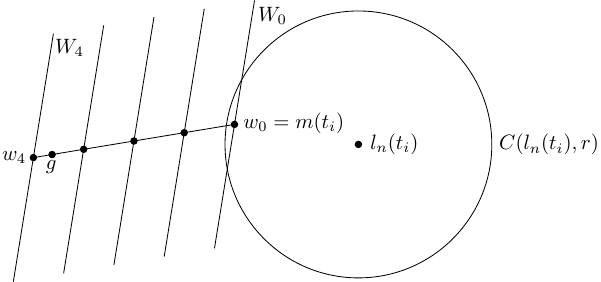}
\caption{The distance between two consecutive of
the parallel lines $W_0,\ldots,W_4$ is at least $\deltaT_n$, which
proves that the man runs from
$\man(t_i)=w_0$ to $w_4$ unless $g$ moves in the meantime.}
\label{claim3proof}
\end{figure}

For the second part of the statement, assume that move $i$ is an escape move.
Let $\goal\mydef \goal(t_i)$.
Let $w_0,\ldots,w_k$ be a sequence of points on the ray from $\man(t_i)$
with direction to $\goal$ such that
$w_0=\man(t_i)$, $\|w_0-w_j\|=j(1+\ee_n)\deltaT_n$,
and $k$ is minimum such that either $\goal\in w_{k-1}w_k$ or
$\goal(t')\neq g$ for some $t'\in [t_{i+k-1},t_{i+k}]$.
See Figure~\ref{claim3proof}.
Let $W_0$ and $W_1$ be the parallel lines defined in case~\ref{choice:b} for
move $i$. We define lines $W_j$ for $j\geq 2$ to be parallel
to $W_0$ and passing through $w_j$. We claim that for any
$j\in\{0,\ldots,k-1\}$, the man moves from $w_j$ to $w_{j+1}$ during
move $i+j$ by making either an escape move or a free move. We prove this by
induction on $j$. The claim holds for $j=0$ by assumption, so assume we have that
$\man(t_{i+j})=w_j$ and that move $i+j-1$ was an escape move or a free move.
Since the distance between consecutive lines $W_j$ and $W_{j+1}$ is
at least $\deltaT_n$, it follows that
$\dist(\lion_n(t_i),W_j)\geq \rD+(j-1)\deltaT_n$ and hence that
$\dist(\lion_n(t_{i+j}),W_j)\geq \rD-\deltaT_n$.
Now, if $\|\man(t_{i+j})-\lion_n(t_{i+j})\|< \rD+\deltaT_n(1+\ee_n)$,
then the lines $W_j$ and $W_{j+1}$ are a witness that move $i+j$ is an
escape move so that the man moves to $w_{j+1}$.
Otherwise, move $i+j$ is a free move, in which case the man likewise moves
to $w_{j+1}$. Finally, since $\goal(t)$ changes or the man reaches $\goal$
during move $i+k$, the statement holds.
\end{proof}

Define $\rhoD'\mydef \rhoD+\rD+(3+\ee_n)\deltaT_n$
and $\tau\mydef 6\pi \rD/\ee_n$.
Informally, $\tau$ is the time it takes for the man to run around the lion $\lion_n$ before he can escape toward $\goal(t)$.
This is made precise by the following claim.

\begin{claim}\label{claim3}
If move $i$ is an avoidance move,
one of the following three events occurs before $\tau$ time
has passed: (i) $\goal(t)$ changes, (ii) $\|\man(t)-\goal(t)\|<\rhoD'$, or (iii) the man makes an escape move.
\end{claim}

\begin{figure}
\centering
\includegraphics[page=9]{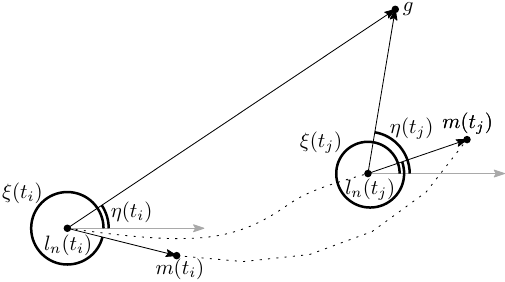}
\caption{The thick arcs show the angles $\eta$ and $\xi$ at two different times $t_i$ and $t_j$, where $t_i<t_j$, when the lion $\lion_n$ and the man $m$ are running along the dotted paths.
Here, $\xi(t_j)>2\pi$, as indicated by the arc making more than one full revolution around $\lion_n(t_j)$.}
\label{etaXiFig}
\end{figure}

\begin{proof}
Let $\goal\mydef \goal(t_i)$.
We first present an informal description of the proof.
Since move $i$ is an avoidance move, we know from Claim~\ref{claim2} that the man keeps making avoidance moves until he makes an escape move.
We show that if the first two events do not occur, he will make an escape move during the interval $[t_i,t_i+\tau]$.
Let $\xi(t)$, respectively~$\eta(t)$, denote the angle of the vector $\overrightarrow{\lion_n(t)\man(t)}$, respectively~$\overrightarrow{\lion_n(t)g}$.
A key observation is that if the difference in these angles is small, then the man makes an escape move since then the lion and the goal $\goal$ are roughly on opposite sides of the man.
Showing that this difference eventually becomes small involves showing that $\xi$ increases by at least $2\pi$ more than $\eta$ after $\tau$ time so that at some point in time $t\in[t_i,t_i+\tau]$, the vectors $\overrightarrow{\lion_n(t)\man(t)}$ and $\overrightarrow{\lion_n(t)g}$ have the same orientation.
By Claim~\ref{claim1}, the lion $\lion_n$ never gets closer than $\rho$ to $g$ which implies that the change in $\eta$ is small in any time interval $[t_j,t_{j+1}]$.
Since the man keeps a minimum distance from the lion, it similarly follows that the change in $\xi$ is small in $[t_j,t_{j+1}]$.
Picking $j$ to be the maximum such that $t_j\le t$ gives $t - t_j\le\deltaT_n$, which implies that the difference in the two angles is small at time $t_j$, at which point the man makes an escape move.
Since $t_j \le t + \tau$, the claim follows.

We now proceed with the proof of the claim.
Since move $i$ is an avoidance move, we know from Claim~\ref{claim2} that the man keeps making avoidance moves until he makes an escape move.
Assume that neither the first nor the second event occurs
before $\tau$ time has passed.
By Claim~\ref{claim1}, we know that
for any $t\in[t_i,t_i+\tau]$, we have that
\begin{align}\label{ineq:liondist}
\|\lion_n(t)-\goal\| & \geq \|\man(t)-\goal\|-\|\man(t)-\lion_n(t)\| \nonumber \\
& \geq \rhoD'-(r+(3+\ee_n)\deltaT_n)= \rhoD. 
\end{align}
Since the distances from $\lion_n$ to $\man$ and to $\goal$ are positive, there exist continuous functions
$\xi,\eta\colon[t_i,t_i+\tau]\longrightarrow\RR$
that measure the angle from $\lion_n$ to $\man$ and to $\goal$,
respectively; see Figure~\ref{etaXiFig}.
These are defined so that for any time $t\in[t_i,t_i+\tau]$, we have
\begin{align*}
\man(t)&=\lion_n(t)+\|\lion_n(t)-\man(t)\|\cdot (\cos \xi(t),\sin \xi(t))\quad\text{and} \\
\goal&=\lion_n(t)+\|\lion_n(t)-\goal\|\cdot (\cos \eta(t),\sin \eta(t)).
\end{align*}

\begin{figure}
\centering
\includegraphics[page=6]{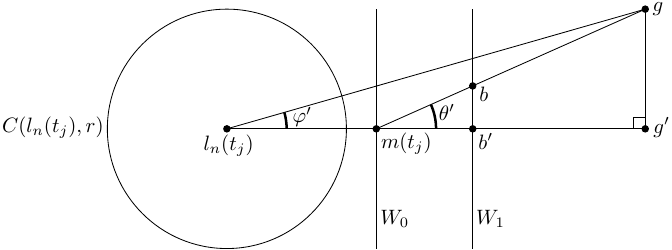}
\caption{Situation of the proof of Claim~\ref{claim3}.}
\label{phiThetaFig}
\end{figure}

Consider an arbitrary time of choice $t_j$ where
$\|\man(t_j)-\lion_n(t_j)\|< \rD+\deltaT_n(1+\ee_n)$,
so that move $j$ is either an escape move or an avoidance move.
Let $\theta'$ be the counterclockwise angle from the direction
$\overrightarrow{\lion_n(t_j)\man(t_j)}$ to $\overrightarrow{\man(t_j)\goal}$.
We now prove that a sufficient condition for move $j$ to be an escape move
is that
$\theta'\leq\theta=\arccos\frac 1{1+\ee_n}$.
See Figure~\ref{phiThetaFig}, and let the point $b$ and the point on the segment $\man(t_j)g$ at distance $\deltaT_n(1+\ee_n)$ from $\man(t_j)$ be as defined in case~\ref{choice:b}.
Consider the two lines $W_0$ and $W_1$
perpendicular to $\lion_n(t_j)\man(t_j)$ through
$\man(t_j)$ and $b$, respectively.
We claim that $W_0$ and $W_1$ are a witness that the man makes an escape move to $b$.
This follows as
\begin{align*}
\dist(\lion_n(t_j),W_0) & =\|\lion_n(t_j)-\man(t_j)\|\geq r-\deltaT_n\quad\text{and}\\
\dist(\lion_n(t_j),W_1) & =\dist(\lion_n(t_j),W_0)+\|\man(t_j)-b'\| \\
& \geq \dist(\lion_n(t_j),W_0)+\cos\theta'\cdot \|\man(t_j)-b\| \\
& \geq \dist(\lion_n(t_j),W_0)+\cos\theta\cdot (1+\ee_n)\deltaT_n \\
& = \dist(\lion_n(t_j),W_0)+\deltaT_n.
\end{align*}

Next, we show that if the difference in the angles
from $\lion_n$ to $\man$ and to $\goal$, respectively, is at most
$\varphi$, then $\theta'\leq\theta$, so the man makes an escape move by the above.
To put it another way, if
\begin{align}\label{xieta}
\lvert \eta(t_j)-\xi(t_j)-2z\pi\rvert\leq \varphi
\end{align}
for some
$z\in\ZZ$, then $\theta'\leq\theta$.
To see this, assume without loss of
generality that inequality~\eqref{xieta} holds for $z=0$,
let $\varphi'\mydef\eta(t_j)-\xi(t_j)$, and consider the case
$0\leq\varphi'\leq\varphi$.
The case $0\geq\varphi'\geq-\varphi$ is analogous.
Let $g'$ be the projection of $g$ onto the line through $\lion_n(t_j)$ and $\man(t_j)$.
We then have
\begin{align*}
\tan \theta'&=
\frac{\|g-g'\|}{\|\man(t_j)-g'\|}=
\frac{\|\lion_n(t_j)-g\|\sin\varphi'}
{\|\lion_n(t_j)-g\|\cos\varphi'-\|\lion_n(t_j)-\man(t_j)\|}.
\end{align*}

Observe that $\tan\theta'$, and hence $\theta'$, are maximum when $\varphi'$ and $\|\lion_n(t_j)-\man(t_j)\|$
are maximum and $\|\lion_n(t_j)-g\|$ is minimum, i.e., when
$\varphi'=\varphi$,
$\|\lion_n(t_j)-\man(t_j)\|=\rD+\deltaT_n$, and
$\|\lion_n(t_j)-g\|=\rhoD$. We therefore get
\begin{align*}
\tan \theta'&\leq
\frac{\rhoD\sin\varphi}
{\rhoD\cos\varphi-(\rD+\deltaT_n)}\leq
\frac{\rhoD\sin\varphi}
{\rhoD\cos\varphi-2\rD}=\tan \theta.
\end{align*}
It now follows that $\theta'\leq\theta$, so indeed, move $j$ is an escape move.

In the following, we prove that there is some time of choice
$t'_j$ in the interval
$[t_i,t_i+\tau]$ for which the condition~\eqref{xieta} is satisfied, i.e.,
condition~\eqref{xieta} is true when $t'_j$ is substituted for $t_j$.

\begin{figure}
\centering
\begin{minipage}[t]{.475\textwidth}
  \centering
  \includegraphics[page=7]{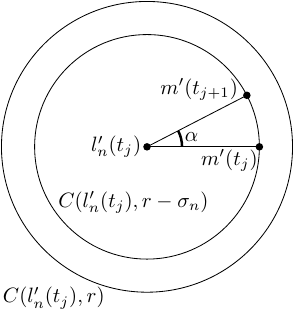}
  \captionof{figure}{The angle $\alpha$ is $\xi(t_{j+1})-\xi(t_j)$.
When
$\|\man'(t_j)-\man'(t_{j+1})\|=(2+\ee_n)\deltaT_n$, then
$\alpha=2\arcsin\frac{(2+\ee_n)\deltaT_n}{2(\rD-\deltaT_n)}$.}
\label{xixiFig}
\end{minipage}\hspace{0.049\textwidth}%
\begin{minipage}[t]{.475\textwidth}
  \centering
  \includegraphics[page=8]{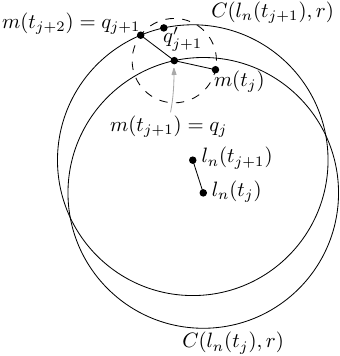}
\captionof{figure}{The man runs the path $\man(t_j)q_jq_{j+1}$. The angle from $q'_{j+1}$ to
$q_{j+1}$ on $C(\lion_n(t_{j+1}),\rD)$ is at least $\ee_n\deltaT_n$.}
\label{angleIncrease}
\end{minipage}\hspace{0.04\textwidth}
\end{figure}

First, we show that for an arbitrary time of choice $t_j$ we have that
\begin{align}\label{maxAngleStep}
\xi(t_{j+1})-\xi(t_j)\leq
2\arcsin\frac{(2+\ee_n)\deltaT_n}{2(\rD-\deltaT_n)}.
\end{align}
To see this, 
define, for $t\in[t_j,t_{j+1}]$,
$$
\lion'_n(t)\mydef\lion_n(t_j)\quad\text{and}\quad
\man'(t)\mydef\man(t)+(\lion_n(t)-\lion_n(t_j)),
$$
i.e., we fix the lion $\lion'_n$ at the point $\lion_n(t_j)$
and let the man $\man'$ run for both so that the segment
$\lion'_n(t)\man'(t)$ is a translation of $\lion_n(t)\man(t)$.
It follows that the man $\man'$ runs at speed at most $2+\ee_n$.
He runs from one point $\man'(t_j)$ to another $\man'(t_{j+1})$, both of which are on or in the exterior of the circle $C(\lion'_n(t_j),\rD-\deltaT_n)$, and $\xi(t_{j+1})-\xi(t_j)$ measures the difference in angles from $\lion'_n(t_j)$ to the two points.
The difference of the angles is therefore maximal if the man $\man'$ runs at full speed $2+\ee_n$ along a straight line segment between two points on $C(\lion'_n(t_j),\rD-\deltaT_n)$.
In other words, $\xi$ cannot increase more on $[t_j,t_{j+1}]$ than in the case that
$\|\lion'_n(t_j)-\man'(t_j)\|=\|\lion'_n(t_{j+1})-\man'(t_{j+1})\|=\rD-\deltaT_n$
and $\|\man'(t_j)-\man'(t_{j+1})\|=(2+\ee_n)\deltaT_n$. From this observation,
inequality~\eqref{maxAngleStep} follows from an elementary argument; see
Figure~\ref{xixiFig}.

We now note that 
\begin{align}\label{maxAngleStep2}
\eta(t_{j+1})-\eta(t_j)\leq
\frac{\deltaT_n}{\rhoD}.
\end{align}
To see this, note that by inequality~\eqref{ineq:liondist}, the lion $\lion_n$ is running on or in the exterior of the circle $C(g,\rho)$ while $\eta$ measures the angle from $\lion_n$ to $g$.
The angle is increasing the most when the lion runs counterclockwise
around $C(g,\rhoD)$ with unit speed, in which case equality holds
in~\eqref{maxAngleStep2}.

Assume now that the moves $i,i+1,\ldots,i+k$ are all avoidance moves and
$t_{i+k}\leq t_i+\tau$.
(See Figure~\ref{angleIncrease}.)
For $j\in\{i,i+1,\ldots,i+k\}$, let
$q_j\mydef \man(t_{j+1})$ be the point to which the man
chooses to run at time $t_j$ as defined in case~\ref{choice:c}.
Let $\xi_j\in[\xi(t_j),\xi(t_j)+\pi]$ be the angle of $q_j$ on
$C(\lion_n(t_j),\rD)$, i.e., the angle such that
$$
q_j=\lion_n(t_j)+\rD\cdot (\cos \xi_j,\sin \xi_j),
$$
as shown in Figure~\ref{figure:xixi}.

\begin{figure}
\centering
\includegraphics[page=10]{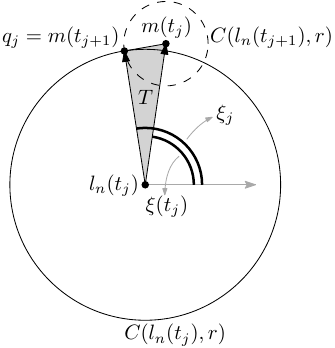}
\caption{The relation between the angles $\xi(t_j)$ and $\xi_j$.}
\label{figure:xixi}
\end{figure}

Let
$q'_{j+1}\mydef q_j+(\lion_n(t_{j+1})-\lion_n(t_j))$ be the point on the circle
$C(\lion_n(t_{j+1}),\rD)$ corresponding to $q_j$ on $C(\lion_n(t_{j}),\rD)$.
Then
\begin{align*}
\|\man(t_{j+1})-q'_{j+1}\| =\|q_j-q'_{j+1}\|=\|\lion_n(t_j)-\lion_n(t_{j+1})\| \leq\deltaT_n.
\end{align*}
Hence,
\begin{align*}
\|q_{j+1}-q'_{j+1}\|\geq\|q_j-q_{j+1}\|-\|q_j-q'_{j+1}\|\geq (1+\ee_n)\deltaT_n-\deltaT_n
=\ee_n\deltaT_n,
\end{align*}
and we get that
$$\xi_{j+1}-\xi_j\geq 2\arcsin\frac {\|q_{j+1}-q'_{j+1}\|}{2r}\geq 2\arcsin\frac {\ee_n\deltaT_n}{2r}>\frac{\ee_n\deltaT_n}{r}$$
for any $j\in\{i,\ldots,i+k-1\}$ and hence that
$\xi_{i+k}-\xi_i> \frac{k\ee_n\deltaT_n}{r}$.

Let $T$ be the triangle formed by the points $\lion_n(t_j)$, $\man(t_j)$, and $q_j=\man(t_{j+1})$ (see Figure~\ref{figure:xixi}), and note that $\xi_j-\xi(t_j)$ is the angle of the corner $\lion_n(t_j)$ in $T$.
The lengths $\|\lion_n(t_j)-q_j\|$ and $\|\man(t_j)-q_j\|$ are fixed at $r$ and $(1+\ee_n)\deltaT_n$, respectively, while the last length $\|\lion_n(t_j)-\man(t_j)\|$ is only known to be in the interval $[r-\deltaT_n,r+\deltaT_n]$.
It follows that $\xi_j-\xi(t_j)$ is maximum when the angle at $\man(t_j)$ is a right angle, in which case
$$
\xi_j-\xi(t_j)\leq \arcsin\frac{(1+\ee_n)\deltaT_n}{r}<\frac{2(1+\ee_n)\deltaT_n}{r}.
$$


Hence we have
\begin{align*}
D&\mydef\left(\xi(t_{i+k})-\xi(t_i)\right)-
\left(\eta(t_{i+k})-\eta(t_i)\right)\\
&> \left(\xi_{i+k}-\frac{2(1+\ee_n)\deltaT_n}{\rD}-\xi_i\right)-
\frac{k\deltaT_n}{\rho}\\
&> \frac{k\ee_n\deltaT_n}{\rD}-\frac{2(1+\ee_n)\deltaT_n}{\rD}-\frac{k\ee_n\deltaT_n}{2\rD}\\
&> \frac{k\ee_n\deltaT_n}{2\rD}-2,
\end{align*}
where the last inequality follows since $\frac{(1+\ee_n)\deltaT_n}{r}<\frac{(3+\ee_n)\deltaT_n}{r}<1$ by the definition of $\deltaT_n$.

Now, if $k\geq\frac{6\pi\rD}{\ee_n\deltaT_n}$, we get
$D>3\pi-2>2\pi$. Hence, after time $\frac{6\pi\rD}{\ee_n\deltaT_n}\cdot\deltaT_n=\tau$,
$\xi$ has increased by at least $2\pi$ more than $\eta$. It follows that at some point
in time $t\in[t_i,t_i+\tau]$ and some $z\in\ZZ$, we have
$$
\lvert\xi(t)-\eta(t)-2z\pi\rvert = 0.
$$
Let $j\in\{i,\ldots,i+k\}$ be maximum such that $t_j\leq t$.
We now prove that condition~\eqref{xieta} is satisfied for the chosen $t_j$. Clearly,
$t-t_j\leq\deltaT_n$. It then follows from inequalities~\eqref{maxAngleStep} and
\eqref{maxAngleStep2} that
\begin{align*}
\lvert\xi(t_j)-\eta(t_j)-2z\pi\rvert
\leq{} & \lvert\xi(t_j)-\xi(t)\rvert+\lvert\xi(t)-\eta(t)-2z\pi\rvert \\
  &+\lvert\eta(t)+2z\pi-\eta(t_j)-2z\pi\rvert
 \\
\leq{} & 2\arcsin\frac{(2+\ee_n)\deltaT_n}{2(\rD-\deltaT_n)}+\frac{\deltaT_n}{\rho}\leq\varphi,
\end{align*}
where the last inequality holds by the choice of $\varphi$.
This finishes the proof of the claim.
\end{proof}

For $i\in\NN_0$, define the \emph{canonical interval} $I_i$
as $I_i\mydef\left[i\deltaT_{n-1}/p,(i+1)\deltaT_{n-1}/p\right)$,
i.e., $I_i$ is the interval of time such that the man would run from the $i$th to the $(i+1)$st milestone on the path defined by the strategy $\manStrat_{n-1}$.
We say that $I_i$ \emph{ends} at time
$t=(i+1)\deltaT_{n-1}/p$.
Note that if $t\in I_i$, then $\goal(t)=\manStrat_{n-1}((i+1)\deltaT_{n-1}/p)$
and $\goal(t)$ changes when $I_i$ ends.

As a consequence of Claim~\ref{claim2} and Claim~\ref{claim3}, we get the following claim, informally stating that once the man gets close to $\goal(t)$ for $t\in I_i$, then he will stay close to $\goal(t)$ until $I_i$ ends.

\begin{claim}\label{claim4}
If $\bar t\in I_i$ and $\|\man(\bar t)-\goal(\bar t)\|\leq\rhoD'$, then for
every $t\geq \bar t$, $t\in I_i$, we have
$$\|\man(t)-\goal(t)\|\leq \rhoD'+(1+\ee_n)\tau.$$
\end{claim}

\begin{proof}
Consider a maximal interval of time $J\subset [\bar t,(i+1)\deltaT_{n-1}/p)\subset I_i$ such that for all $t\in J$, we have $\|\man(t)-\goal(t)\|> \rhoD'$.
We show that during $J$, the man gets to distance at most $\rhoD'+(1+\ee_n)\tau$ from $g(t)$, and the claim thus follows.
Let $J=(t'_0,t'_1)$.
At time $t'_0$, the man must be making an avoidance move since his distance to $g(t)$ is increasing beyond $\rhoD'$.
By Claim~\ref{claim3}, one of the following three event occurs before time $\tau$ has passed after $t'_0$: (i) $I_i$ ends, (ii) $\|m(t)-g(t)\|<\rhoD'$, or (iii) the man makes an escape move.
In case (i) and (ii), it trivially follows that the length of $J$ is at most $\tau$.
Hence, the man can get to distance at most $\rhoD'+(1+\ee_n)\tau$ from $g(t)$ during $J$.
In case (iii), let the escape move be at the time of choice $t_j$.
We then have that in the interval $[t'_0,t_j]$, the man gets to distance at most $\rhoD'+(1+\ee_n)\tau$ from $g(t)$.
By Claim~\ref{claim2}, he then makes escape moves or free moves until he reaches $g(t)$ or $I_i$ ends.
Therefore, his maximum distance to $g(t)$ during $J$ is also $\rhoD'+(1+\ee_n)\tau$ in this case.
\end{proof}

Finally, the following claim informally states that the man is always close to where he would be according to the strategy $M_{n-1}$.

\begin{claim}\label{claim5}
For any $i\in\NN_0$ and at any time during the canonical interval $I_i$, the man is at distance at most $\rhoD'+2(1+\ee_n)\tau$ away from the segment $\manStrat_{n-1}(I_i)$ and when $I_i$ ends, the man is within distance $\rhoD'+(1+\ee_n)\tau$ from the endpoint $\manStrat_{n-1}((i+1)\deltaT_{n-1}/p)$ of the segment.
\end{claim}

\begin{proof}
We prove the claim by induction on $i$.
To easily handle the base case, we introduce
an auxiliary canonical interval $I_{-1}=[-\deltaT_{n-1}/p,0)$ and assume that the lions and the man (according to both strategies $\manStrat_{n-1}$ and $\manStrat_n$) are standing at their initial positions during all of $I_{-1}$.
The statement clearly holds for $i=-1$.

Assume inductively that the statement holds for $I_{i-1}$ and consider the interval $I_i$.
Let $\goal\mydef \manStrat_{n-1}((i+1)\deltaT_{n-1}/p)$.
During $I_i$, the man can run the distance
\begin{align*}
\deltaT_{n-1}/p\cdot (1+\ee_n)={} & \deltaT_{n-1}/p\cdot (1+\ee_{n-1})+\deltaT_{n-1}/p\cdot (\ee_n-\ee_{n-1}) \\
\geq{} & \deltaT_{n-1}/p\cdot (1+\ee_{n-1}) + \rhoD+2\rD+3(1+\ee_n)\tau \\
>{} & \deltaT_{n-1}/p\cdot (1+\ee_{n-1}) + \rhoD'+3(1+\ee_n)\tau,
\end{align*}
where the two inequalities follow from the definitions of $r$ and $\deltaT_{n-1}$, respectively.

By the induction hypothesis, the man is within a distance of
$\rhoD'+(1+\ee_n)\tau$ from $\manStrat_{n-1}(i\deltaT_{n-1}/p)$
at time $i\deltaT_{n-1}/p$. Thus, his distance to $\goal$ at the beginning of
interval $I_i$ is
$\|\man(i\deltaT_{n-1}/p)-g\|\leq \deltaT_{n-1}/p\cdot (1+\ee_{n-1})+\rhoD'+(1+\ee_n)\tau$,
where $\deltaT_{n-1}/p\cdot (1+\ee_{n-1})$ is the length of the interval
$M_{n-1}(I_i)$.

If the man does not make any avoidance moves during $I_i$, he runs
straight to $\goal$, so it follows that he reaches $\goal$ at the latest at time
\begin{align}\label{ineq:timeToGoal}
\frac{\|\man(i\deltaT_{n-1}/p)-g\|}{1+\ee_n}={} & \deltaT_{n-1}/p+\frac{\|\man(i\deltaT_{n-1}/p)-g\|-\deltaT_{n-1}/p\cdot (1+\ee_n)}{1+\ee_n} \nonumber \\
<{} & \deltaT_{n-1}/p+\frac{(1+\ee_n)\tau-3(1+\ee_n)\tau}{1+\ee_n} \nonumber \\
={} & \deltaT_{n-1}/p-2\tau.
\end{align}
In the remaining part of $I_i$ after reaching $\goal$, he will be within distance $\deltaT_n(1+\ee_n)$ from $g$ (possibly doing a sequence of free moves in each of which he passes over $\goal$), so the statement is true in this case.

\begin{figure}
\centering
\includegraphics[page=11]{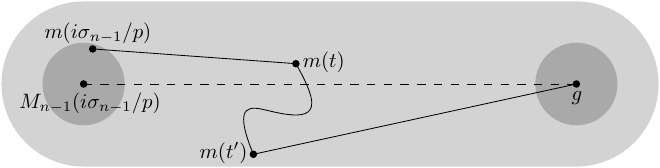}
\caption{Situation of the proof of Claim~\ref{claim5}.
The dark gray disks are the points at distance at most $\rhoD'+(1+\ee_n)\tau$ from $\manStrat_{n-1}(i\deltaT_{n-1}/p)$ and $\goal$.
The light gray hippodrome consists of the points at distance at most $\rhoD'+2(1+\ee_n)\tau$ from the segment $\manStrat_{n-1}(I_i)$.}
\label{figure:claim5}
\end{figure}

Otherwise, let $t\in I_i$ be the first time of choice at which he makes an avoidance move during $I_i$; see Figure~\ref{figure:claim5}.
If $t\geq (i+1)\deltaT_{n-1}/p-2\tau$, then he already reached $\goal$ before time $t$, and the statement follows from Claim~\ref{claim4}, so assume that $t<(i+1)\deltaT_{n-1}/p-2\tau$.
Hence, Claim~\ref{claim3} gives that at some time $t'\leq t+\tau$, either
\begin{enumerate}[leftmargin=1.6em]
\item\label{claim5:1}
the man gets within a distance of $\rhoD'$
from $\goal$, or

\item\label{claim5:2}
he does an escape move.
\end{enumerate}

We first prove that in the interval $[t,t']$,
the distance from the man to the segment $\manStrat_{n-1}(I_i)$
is at most $\rhoD'+2(1+\ee_n)\tau$.
To this end, note that during the preceding interval $[i\deltaT_{n-1}/p,t)$ of $I_i$, he makes no avoidance moves and thus only gets closer to the segment.
His initial distance is $\rhoD'+(1+\ee_n)\tau$ by the induction hypothesis, so that is also a bound on his distance up to time $t$.
Thus, since $t'\leq t+\tau$, his distance at time $t'$ can be at most $\rhoD'+2(1+\ee_n)\tau$.

It remains to be proved that the man stays within distance
$\rhoD'+2(1+\ee_n)\tau$ from $\manStrat_{n-1}(I_i)$ after
time $t'$ and that he is at distance at most
$\rhoD'+(1+\ee_n)\tau$ from $\goal$ at time $(i+1)\deltaT_{n-1}/p$.
If we are in case~\ref{claim5:1}, the statement follows from
Claim~\ref{claim4}, so assume case~\ref{claim5:2}.

By Claim~\ref{claim2}, the man will run directly toward $\goal$ after time $t'$ until
(i) he reaches $\goal$ or (ii) $I_i$ ends.
We now claim that we will always be in case (i), i.e., he always reaches $\goal$ before $I_i$ ends.
The length of the path he has to run in order to reach $\goal$ is at most
$$
L\mydef \|\man(i\deltaT_{i-1})-\man(t)\|+\tau(1+\ee_n)+\|\man(t')-g\|,
$$
where $\tau(1+\ee_n)$ is a bound on the length of the path he runs during $[t,t')$ while making avoidance moves.
The triangle inequality gives
$$\|\man(t')-g\|\leq \|\man(t')-\man(t)\|+\|\man(t)-g\|\leq \tau(1+\ee_n)+\|\man(t)-g\|.$$
We therefore get
$$
L\leq \|\man(i\deltaT_{i-1})-g\|+2\tau(1+\ee_n).
$$

We now get from inequality~\eqref{ineq:timeToGoal} that the time it takes the man to run to $g$ is at most
$$
\frac L{1+\ee_n}\leq \frac{\|\man(i\deltaT_{i-1})-g\|}{1+\ee_n}+2\tau=\deltaT_{n-1}/p,
$$
so he indeed reaches $\goal$ before $I_i$ ends.
Note that while he runs to $\goal$, he stays within distance $\rhoD'+2(1+\ee_n)\tau$ from segment $\manStrat_{n-1}(I_i)$.
After he has reached $\goal$, we get by Claim~\ref{claim4} that he stays within distance $\rhoD'+(1+\ee_n)\tau$ from $\goal$ for the rest of $I_i$.
This finishes the proof.
\end{proof}

We are now ready to finish our proof of Lemma~\ref{LEMMA:MAIN}.
We define the safety distance to lion $\lion_n$ as $c_n\mydef \rD-(3+\ee_n)\deltaT_n$.
By Claim~\ref{claim1}, we know that the safety distance is maintained.

We now give a bound on the distance $\|M_{n-1}(t)-M_n(t)\|$ at any time $t$.
Suppose that $t\in I_i$.
Claim~\ref{claim5} says that during interval $I_i$, the distance from the man to the segment $\manStrat_{n-1}(I_i)$ is at most
\begin{align}\label{eq:boundingSegDist}
\rhoD'+2(1+\ee_n)\tau<
\rhoD+2\rD+2(1+\ee_n)\tau\leq \deriv_n/2,
\end{align}
where the first inequality follows from the choices of $\rho'$ and $\deltaT_n$ and the second from the definition of $r$.
Since $M_{n-1}(t)$ is a point on the segment $\manStrat_{n-1}(I_i)$ of length $\deltaT_{n-1}/p\cdot (1+\ee_{n-1})\leq\deriv_n/2$, the bound~\eqref{eq:boundingSegDist} implies that
\begin{align}\label{eq:manmanDist}
\|M_{n-1}(t)-M_n(t)\|\leq \deriv_n/2+\deriv_n/2=\deriv_n.
\end{align}
This proves the lemma.

\end{proof}

\section{Concluding Remarks}
We have shown that a fast man can survive any finite number of lions, and that the lions do not also have a winning strategy since the man's strategy is locally finite.
Furthermore, the man has a winning strategy $\manStrat_\infty$ against any countably infinite set of lions.
However, this strategy is in general not locally finite.
Indeed, if it was, there would be $\delta>0$ such that the man's path according to $\manStrat_\infty$ up to time $\delta$ was already determined at time $0$.
If the lions' start points were dense in the plane, a lion sufficiently close to the point $\manStrat_\infty(\delta)$ could catch the man at time $\delta$ by running to that point, and hence $\manStrat_\infty$ would not be winning.
Therefore, $\manStrat_\infty$ cannot be locally finite, so the lions might also have a winning strategy, but we have not been able to find one.

We finally note that for the man to have a winning strategy, it is necessary that there are only \emph{countably} many lions.
Indeed, the man cannot win if there is a lion at every point of the plane except for his starting point.

\section*{Acknowledgments}

We thank the reviewers of the published paper~\cite{abrahamsen2020escaping} for their comments and suggestions that improved the paper a lot.
We thank Mehdi Salimi for making us aware of the papers~\cite{chernousko1976problem,ibragimov2012evasion}.





\newcommand{\mdoi}[1]{\href{http://dx.doi.org/#1}{\texttt{doi.org/#1}}}

\bibliographystyle{plain}



\end{document}